\documentclass[11pt]{article}
\usepackage{mathrsfs}
\usepackage{latexsym}
\usepackage{amsmath,amssymb}
\usepackage[pdftex]{hyperref}
\usepackage{amsthm}
\usepackage{amsfonts,cite}
\usepackage[usenames]{color}
\usepackage{amssymb}
\usepackage{graphicx}
\usepackage{amsmath}
\usepackage{amsfonts}
\usepackage{amsthm}
\usepackage{mathrsfs}
\usepackage{dsfont}
\newcommand\be{\begin{equation}}
\newcommand\ee{\end{equation}}
\newcommand\ber{\begin{eqnarray}}
\newcommand\eer{\end{eqnarray}}
\newcommand\berr{\begin{eqnarray*}}
\newcommand\eerr{\end{eqnarray*}}
\newcommand\bea{\begin{eqnarray}}
\newcommand\eea{\end{eqnarray}}
\newcommand\ba{\begin{array}}
\newcommand\ea{\end{array}}
\newcommand\bfR{\mathbb{R}}
\newcommand\dd{\mathrm{d}}
\newcommand\lm{\lambda}\newcommand\om{\omega}
\newcommand\ii{\mbox{i}}
\newcommand\re{\mathrm{e}}\newcommand\e{\mathrm{e}}
\newcommand\eq{\eqref}\newcommand\lb{\label}

\newcommand\ri{\mathrm{i}}
\newcommand\pa{\partial}

\newcommand{\nn}{\nonumber}

\newcommand{\bfZ}{{\Bbb Z}}

\newcommand\vp{\varphi}
\newcommand{\vep}{\varepsilon}

\newcommand{\bi}{\begin{itemize} }
  \newcommand{\ei}{\end{itemize} }

\newtheoremstyle{mythm}{1.5ex plus 1ex minus .2ex}{1.5ex plus 1ex
minus .2ex}{\kai}{\parindent}{\song\bfseries}{}{1em}{}
\numberwithin{equation}{section}\numberwithin{figure}{section}

\newtheorem{theorem}{Theorem}[section]
\newtheorem{lemma}{Lemma}[section]

\allowdisplaybreaks[4]
\begin{document}
\title{Bogomol'nyi Equations and Coexistence of Vortices\\ and Antivortices in Generalized Abelian Higgs Theories}
\author{Aonan Xu \\ School of Mathematics and Statistics\\Henan University\\
Kaifeng, Henan 475004, P. R.  China\\ \\ Yisong Yang \\ Courant Institute of Mathematical Sciences\\ New York University, New York }

\date{}
\maketitle

\begin{abstract}
We derive the Bogomol'nyi equations in generalized Abelian Higgs theories which allow the coexistence of vortices and antivortices
over a compact Riemann surface or the full plane.
In the compact surface situation, we obtain a necessary and sufficient condition for the existence of a unique solution describing a system of coexisting vortices and antivortices.
In the full-plane situation, we prove the existence of a unique solution representing an arbitrary distribution of vortices and antivortices and obtain sharp asymptotic behavior of the solution near infinity. These solutions carry quantized magnetic fluxes and energies explicitly expressed in terms of the numbers of vortices and antivortices topologically characterized by the first Chern and
Thom classes.

\medskip


\medskip


\medskip


\end{abstract}

%{\bf PACS numbers.} 02.30.Jr, 02.30.Xx, 11.15.-q, 74.25.Ha
\medskip

%{\bf MSC numbers.} 35J50, 53C43, 81T13
\section{Introduction}
The concept of vortices in quantum field theory has its origins in the study of type II superconductors pioneered by Abrikosov \cite{Ab}. These vortices arise as stable, localized topological defects due to the interplay of spontaneous symmetry breaking and the topological properties of the underlying gauge theory. This phenomenon has far-reaching implications in condensed matter physics, high-energy physics, and mathematics (particularly in differential geometry and topology).
In gauge theories, vortex solutions appear naturally when a gauge field couples to a charged scalar field with spontaneous symmetry breaking. This is exemplified in the Abelian-Higgs model, where solutions to the coupled equations of motion lead to stable, localized vortex configurations. Mathematically, these vortices arise as
solutions to the Bogomol'nyi \cite{Bo} type equations in the self-dual limit
 classified by topological invariants, for example, the first Chern number in fiber bundle language, and
 are closely related to the study of holomorphic vector bundles, moduli spaces, and geometric analysis.
Physically, vortices also appear as cosmic strings, which are relevant in early universe models.
Moreover, recent studies suggest that vortices play a crucial role in understanding the confinement mechanism of monopoles in non-Abelian gauge theories. This has direct implications for the quark confinement puzzle in quantum chromodynamics, where vortex condensation is proposed as a mechanism leading to a linear confinement mechanism via dual colored superconductivity
\cite{Konishi1,Konishi2,ShY1,ShY2}.

In some gauge field theories, vortices and antivortices can coexist, leading to rich physical structures and interactions \cite{Bot,Chae,Can,Her,Ruo}. Their mutual attraction or repulsion depends on the specific topological charge and gauge interactions, influencing phase transitions and dynamical properties of the field configurations. The study of vortex-antivortex pairs contributes to understanding duality transformations and the non-perturbative aspects of gauge theories.
On the other hand, however,
analyzing the governing equations for a vortex-antivortex system is mathematically challenging due to several difficulties including
nonlinearity of the field equations,
coupling of multiple fields,
topological constraints associated with
vortices and antivortices having opposite winding numbers and leading to topological charge cancellation, and
long-range interactions between
vortex-antivortex pairs.
Due to these mathematical difficulties, most studies rely on numerical simulations or approximations such as the Bogomol'nyi limit, to understand vortex-antivortex dynamics. In fact, these approaches are also what have been taken in the studies of single-species vortex systems in various theories. The goal of this study is to
obtain a broad family of Bogomol'nyi equations in generalized Abelian--Higgs theories that allow the coexistence of an arbitrarily distributed system of vortices and antivortices.

This work is initiated from two field-theoretical origins. The first origin is the gauged $\sigma$-model of Schroers \cite{Sch1,Sch2} in which the complex Higgs field representation of the underlying
$O(3)$-map leads to the coexistence of the zeros and poles of the complex field which give rise to the concentrated magnetic field and energy density and opposite local and global
topological properties \cite{Sch1,Sch2,SSY,Y1,Y2}. The second origin is the generalized Abelian Higgs theories of Lohe \cite{Lo1,Lo2} which enable one to achieve the same
Bogomol'nyi reduction for systems with general Higgs potentials \cite{Adam,Alo,Baz}. In particular, in such a formalism, one may regard the gauged $\sigma$-model as a special example. It is this connection that motivates our work here which aims to obtain the Bogomol'nyi equations in generalized Abelian Higgs theories, allowing the coexistence of vortices and antivortices with topological characteristics, in a unified treatment.

An outline of the main content of the paper is as follows. In Section 2, we introduce our generalized Abelian Higgs theories. In Section 3, we derive the generalized Bogomol'nyi equations
and show that the solutions are characterized by the first Chern class and the Thom class of the underlying Hermitian line bundle. In Section 4, we come up with the nonlinear elliptic equation
that governs the vortices and antivortices.
This equation will then be studied in Sections 5 and 6 and sharp existence and uniqueness theorems for the solutions representing coexisting vortices and antivortices over a compact Riemann surface
and the full plane will be established. In the compact surface situation, we use a fixed-point theory argument to prove the existence of solutions. In the full plane situation, we use
a sub- and supersolution method.
In Section 7, we conclude with some remarks.

\section{Generalized Abelian Higgs theory over a Riemann surface}

The field-theoretic problem we consider here has a characteristic feature of being topological which can be formulated in the framework of a Hermitian line bundle $L$ over a Riemann
surface $S$. Within this setting, let $u$ be a cross-section resembling a complex scalar field, or the Higgs field, and $A$ a connection 1-form which is real-valued and defines the connection
or the covariant differentiation on $u$ by
\be
D_A u=\dd u-\ii Au,
\ee
giving rise to the curvature 2-form $F_A=\dd A$. Use $*$ to denote the Hodge dual operator. Then the Hamiltonian energy density of the Abelian Higgs theory, in the critical BPS coupling
\cite{Bo,PS,R}, is
\be\lb{2.2}
{\cal H}(A,u)=\frac12 *(F_A\wedge *F_A)+\frac12*(D_A u\wedge *\overline{D_A u})+\frac18 (1-|u|^2)^2.
\ee
A complete understanding of the multivortex solutions topologically characterized by the first Chern class of the bundle $L\to S$ given by
\be
c_1(F)=\frac1{2\pi}\int_S F_A =M,\quad M\in\bfZ,
\ee
and minimizing the total energy functional associated with \eq{2.2}, has been achieved \cite{Bra,JT,N1,N2}, where the topological  integer $M$ is the total vortex number which counts
the algebraic number
of zeros of the field $u$.

In \cite{SSY,Y1,Y2}, the Abelian Higgs theory \eq{2.2} is extended such that the Hamiltonian energy density assumes the form
\be\lb{2.4}
{\cal H}(A,u)=\frac12 *(F_A\wedge *F_A)+\frac2{(1+|u|^2)^2}*(D_A u\wedge *\overline{D_A u})+\frac12\left(\frac{1-|u|^2}{1+|u|^2}\right)^2.
\ee
This theory originates from the gauged $\sigma$-model pioneered by Schroers \cite{Sch1,Sch2} and enjoys many distinguishing features, including
\begin{enumerate}

\item[(i)] The energy density \eq{2.4} formally returns to \eq{2.2} in the limit $|u|^2\to 1$.

\item[(ii)] The energy density \eq{2.4} is finite at the poles of $u$ as well such that the solutions of the theory are characterized by both poles and zeros of the field $u$. The zeros and poles
of $u$ give rise to the spots where $F_A$, or $*F_A$ rather, equivalently, viewed as the vorticity fields attains its global maxima and minima and thus are identified as vortices and antivortices. In this situation, the
energy of the system is proportional to the sum of the number of vortices and number of antivortices.

\item[(iii)] The model \eq{2.4} accommodates {\em coexisting} vortices and antivortices as well as {\em single-species} vortices or antivortices.

\item[(iv)] As in the Abelian Higgs theory \eq{2.2}, this theory also possesses a spontaneously broken $U(1)$ symmetry at $|u|^2=1$. In fact, such a symmetry is enlarged slightly into
a $U(1)\times{\bfZ}_2$ type of the form:
\be
(A,u)\mapsto (A+\dd\chi,u\e^{\ii\chi}),\quad (A,u)\mapsto\left(-A,\frac1u\right),
\ee
where $\chi$ is a real scalar field.

\item[(v)] In order to take account of antivortices, another topological invariant, the Thom class, emerges together with the first Chern class as characterizations of the solutions.

\item[(vi)] In the compact situation, it is known that a necessary and sufficient condition for existence of an $M$-vortex solution for the theory \eq{2.2} is that $M$ stays below an explicit bound, called the Bradlow
bound or limit \cite{Gud,Manton}. For the theory \eq{2.4}, on the other hand, the condition is replaced by that the difference of the number of vortices, $M$, and the number of antivortices, $N$, stays below
an explicit bound \cite{SSY}. That is, in this situation, both $M$ and $N$ are allowed to be arbitrarily large, as far as $|M-N|$ stays below a certain bound. In particular, for \eq{2.2}, the
associated energy is
bounded by the total surface area, but for \eq{2.4}, the associated energy is not bounded by the total surface area since it is proportional to $M+N$ which can be arbitrarily large regardless what
the surface area is.
\end{enumerate}

Note that the line bundle formalism \cite{GH} is naturally suited to describe a system of vortices: local vortex configurations are defined over coordinate patches and glued together via transition functions, with their global compatibility governed by the first Chern class. This ensures topological consistency of the phase winding across the entire surface.

It is interesting to note that \eq{2.4} is a special example of the generalized Abelian Higgs theory developed much earlier by Lohe \cite{Lo1} (see also \cite{Lo2,Lo3}) with the Hamiltonian energy density
\be\lb{2.6}
{\cal H}(A,u)=\frac12 *(F_A\wedge *F_A)+\frac12F(|u|^2)*(D_A u\wedge *\overline{D_A u})+\frac12w(|u|^2)^2,
\ee
where $F$ and $w$ are some functions to be determined, which are not free to pick but mutually related in a specific manner in order to accommodate multivortices. Varying $A$ and $u$ in
\eq{2.6}, we obtain the associated Euler--Lagrange equations of the theory:
\bea
\dd * F_A&=&- F(|u|^2)\frac\ii2(u\overline{D_A}u-\overline{u}D_Au),\lb{a2.7}\\
D_A*(F(|u|^2) D_A u)&=&(F'(|u|^2) *(D_A u\wedge *\overline{D_A u})+2w(|u|^2)w'(|u|^2))u.\lb{a2.8}
\eea

Since in \eq{2.6} the last
term represents the potential density function
\be
V(|u|^2)\equiv \frac12 w(|u|^2)^2,
\ee
we see that the presence of spontaneous broken symmetry as that in \eq{2.2} or \eq{2.4} imposes the condition
\be\lb{xx2.10}
w(1)=0,
\ee
say,  such that the vacuum manifold is realized as $\{u=\e^{\ii\theta}\,|\,\theta\in\bfR\}$. With the generalized energy density \eq{2.6}, we have flexibility in choosing the potential density
function. This freedom is often advantageous and useful in applications \cite{Av,Cas,Ghosh,Izq}.

Motivated by these earlier works and applications, we study the coexistence of vortices and antivortices of the generalized Abelian Higgs theory \eq{2.6} by exploiting its Bogomol'nyi equations
structure \cite{Bo,PS} and the underlying nonlinear partial differential equations.

\section{Bogomol'nyi equations, vortex numbers, and minimum energy}

First recall the identities
\bea
|D_A u|^2&=&*(D_A u\wedge *\overline{D_Au}),\\
D_A u\wedge*\overline{D_A u}+*D_A u\wedge\overline{D_A u}&=&(D_A u\pm\ii*D_Au)\wedge *\overline{(D_Au\pm\ii*D_Au)}\nn\\
&&\pm\ii(D_Au\wedge\overline{D_Au}-*D_A u\wedge *\overline{D_Au}).
\eea
Hence we can rewrite \eq{2.6} as
\bea\label{2.8}
\mathcal{H}(A,u)
&= & \frac{1}{2}\left| F_A \mp *w(|u|^2)\right|^{2} \pm * F_{A} w(|u|^2)\nn\\
&&+\frac14{F(|u|^2)} \left(\left| D_A u\pm\ii*D_Au\right|^{2} \pm \mathrm{i}*(D_Au\wedge\overline{D_Au}-*D_A u\wedge *\overline{D_Au})\right)\nn \\
&= & \frac{1}{2}\left| F_A \mp *w(|u|^2)\right|^{2} + \frac14 F(|u|^2)\left| D_A u\pm\ii*D_Au\right|^{2} \pm  * F_A \nn\\
& &\pm \left(( w(|u|^2)-1)* F_A+\frac{\mathrm{i}}4F(|u|^2) *(D_Au\wedge\overline{D_Au}-*D_A u\wedge *\overline{D_Au}) \right),\quad\quad
\eea
 On the right-hand side above, the first two terms are quadratic and the third term gives rise to the first Chern class as mentioned. So it remains to recognize the last few terms.

For this purpose, we introduce a new current density
\be\label{2.9}
 J= \mathrm{i}f(|u|^2)\left(u \overline{D_A u}-\overline{u} D_A u\right).
\end{equation}
Applying the commutator relation
\begin{equation}
[D_{i},D_{j}]u=(D_{i}D_{j}-D_{j}D_{i})u=-\mathrm{i}F_{ij}u
\end{equation}
in local coordinates to \eqref{2.9}, we obtain
\begin{equation}\label{2.10}
\dd J=-2|u|^2 f(|u|^2) F_A+\mathrm{i}\left(f(|u|^2)+f^{\prime}(|u|^2) |u|^2\right) \left(D_Au\wedge\overline{D_Au}-*D_A u\wedge *\overline{D_Au}\right).
\end{equation}
Equalizing \eqref{2.10} with the last term on the right-hand side of \eqref{2.8}, we derive the identities
\begin{equation}\label{2.11}
   w(s)-1=-2sf(s), \quad F(s)=4\left(f(s)+f^{\prime}(s) s\right).
\end{equation}
Thus we conclude that there holds
\begin{equation}\label{2.12}
w(s)=\frac12 \int_{s}^{1} F(\rho) \mathrm{d} \rho, \quad \text { or equivalently, } \quad F(s)=-2 w^{\prime}(s), \quad w(1)=0.
\end{equation}
From \eqref{2.12}, we observe that we arrive at the normalization condition \eq{xx2.10} such that  the boundary condition
\begin{equation}\label{a2.12b}
    \left|u\right|^{2}=1
\end{equation}
is  imposed as the vacuum manifold as desired. Combining \eqref{a2.12b} with the first equation of \eqref{2.11}, we obtain the condition
\begin{equation}\lb{aa3.10}
    f(1)=\frac{1}{2}.
\end{equation}
Besides, the structure of the energy density \eq{2.6}  indicates that $F(s)\geq0$ which implicates in view of \eq{2.11}  that $f(s)$ must satisfy the condition
\be\lb{aa3.11}
(sf(s))'\geq0.
\ee
Consequently, we see that the function $w(s)$ in \eq{2.11} is monotone decreasing. Furthermore, since we aim to develop a field theory that will accommodate both vortices and antivortices,
which are to be realized by the zeros and poles of $u$, we should require that $w(s)$ be regular at both $s=0$ and $s=\infty$. The latter condition implies that the transformed function
\be
h(t)=w\left(\frac1t\right)
\ee
is differentiable for $t$ near $t=0$. This property leads us to the property
\be\lb{aa3.13}
w(s)=h(0)+\frac{h'(0)}s+\mbox{O}(s^{-2}),\quad s\gg1.
\ee
On the other hand, the first relation in \eq{2.11} gives us $w(0)=1$.  Therefore, if  vortices and antivortices are treated on equal footing energetically, the potential energy density
should assume the same value at $s=0$ and $s=\infty$ which leads to the condition $w(\infty)=-1$, giving rise to $h(0)=-1$ in \eq{aa3.13}. Thus, we obtain
\be\lb{aa3.14}
\lim_{s\to\infty} sf(s)=1.
\ee
The properties \eq{aa3.11} and \eq{aa3.14} are natural conditions that the function $f(s)$ linking the functions $F(s)$ and $w(s)$ should satisfy, which will be observed in our subsequent
work.

As an example, note that the special situation
\be\lb{a3.11}
 f(s)=\frac1{1+s},\quad w(s)=\frac{1-s}{1+s},\quad F(s)=\frac4{(1+s)^2}
\ee
returns to the classical model \eq{2.4}

Summarizing these results, we rewrite \eq{2.8} as
\be\lb{a3.12}
{\cal H}(A,u)=\frac{1}{2}\left| F_A \mp *w(|u|^2)\right|^{2} + \frac14 F(|u|^2)\left| D_A u\pm\ii*D_Au\right|^{2} \pm * \left(  F_A+\dd J\right).
\ee
Hence we arrive at the following energy lower bound
\bea\lb{a3.13}
E(A,u)&=&\int_{S}*{\cal H}(A,u)\nn\\
&\geq& |T|,
\eea
where the quantity
\be
T=\int_{S}\left(F_A+ \dd J\right)
\ee
is topological which will be shown to be proportional to the sum of vortex and antivortex numbers,  $M$ and $N$. That is, $T=2\pi (M+N)$, to be precise.

Thus we see that the topological energy lower bound \eq{a3.13} is attained if and only if the two quadratic terms in \eqref{a3.12} both vanish:
\bea
*F_A= \pm  w(\left|u\right|^{2}), \lb{2.13}\\
D_{A}u\pm  \mathrm{i}*D_{A}u = 0. \lb{2.14}
\eea
These are the Bogomol'nyi equations we set forth to derive.
Since the solutions of \eq{2.13} and \eq{2.14} are the minimizers of the energy functional \eq{a3.13}, stratified by the stated topological lower bounds, they automatically solve the original Euler--Lagrange equations
\eq{a2.7} and \eq{a2.8}. These Bogomol'nyi equations are first-order differential equations and thus a significant reduction of the original second-order equations, \eq{a2.7} and \eq{a2.8}.
In the subsequent study, we shall focus on \eq{2.13} and \eq{2.14}.
In symplectic geometry, this type of equations belong to a family of equations referred to as Hitchin's equations \cite{D,H1,H2}
such that the first equation relates the curvature $F_A$ of the line bundle connection $A$ to the moment map of the $U(1)$ gauge group action
on the pair $(A,u)$ where $u$ is a meromorphic section. In \cite{RS}, Rom\~{a}o and Speight present a study on the geometry of
the moduli space of vortices and antivortices associated with the classical model \eq{2.4} or \eq{a3.11}. It would be interesting to investigate that,
to what extent, their work may be extended to the generalized system of equations \eq{2.13} and \eq{2.14}, and whether the flexibility of the nonlinearity
in $w(s)$ can be exploited advantageously. In Section \ref{sec5}, we further comment on the possible
physical meaning and applications of the solutions to these equations, especially with regard to spontaneous vortex formation and the associated
symmetry breaking phenomenon by an
externally exerted magnetic field and a type-II superconductivity interpretation in the general setting. In short, in various situations, the freedom in choosing the nonlinear function
$w$ in \eq{2.13} should enable the treatment of a wider range of geometric and physical applications.

\section{Elliptic equation governing vortices and antivortices}

In local isothermal coordinates $(x_1,x_2)$ over $S$ in which the line element of $S$ is given in a conformally flat manner by $\dd s^2=\eta(x_1,x_2)(\dd x_1^2+\dd x_2^2)$ where $\eta(x_1,x_2)$ is a smooth positive-valued function, the equation \eq{2.14} takes the form
\be
(\pa_1\pm\ii\pa_2) u=\ii(A_1\pm\ii A_2) u,\quad  A=A_1\dd x_1+A_2\dd x_2.
\ee
Thus, applying the $\overline{\pa}$-Poincar\'{e} lemma as in \cite{JT}, it is seen that, the zeros and poles of $u$ are isolated and of integer multiplicities. More precisely, if we use $z_0\in S$
to denote a zero or pole of $u$ in $S$, then there is an integer $m\geq1$ such that
\be\lb{a4.2}
u(z)=\om(z) (z-z_0)^{\pm m},
\ee
where the plus or minus sign in front of the integer $m$ depends on whether $z_0$ is a zero or a pole and $\om(z)$ is a smooth nonvanishing function for $z=x_1+\ii x_2$ near $z_0$.

With the characterization \eq{a4.2},  assume that the sets of zeros and poles of $u$ are
\begin{equation}\label{a4.3}
    Q=\left\{q_{1},q_{2},...,q_{M}\right\},\quad P=\left\{p_{1},p_{2},...,p_{N}\right\},
\end{equation}
respectively, with repetitions counting for multiplicities.  These zeros and poles give rise to vortices and antivortices. For an illustration of this, let us consider the classical model \eq{a3.11}. Then
\eq{2.13} reads
\be
*F_A=  \frac{1-|u|^2}{1+|u|^2},
\ee
with the plus sign taken. If  $*F_A$ is considered as a vorticity field, we see that it enjoys the bounds $-1\leq *F_A\leq 1$ and that it attains its maxima 1 and minima $-1$ at the zeros and
poles of $u$, respectively. In other words, the zeros and poles of $u$ are where the vorticity field concentrates such that its associated flow velocity field $A=(A_1,A_2)$ gives rise to
flow-lines with opposite windings, hereby the names, vortices and antivortices.

Note that we can resolve \eq{2.14} to obtain the representation
\be
*F_A=\mp\frac12\Delta \ln|u|^2,\quad\mbox{away from the zeros and poles of $u$},
\ee
where
\be
\Delta v=\frac{1}{\sqrt{|g|} } \partial_{i}(g^{ij}\sqrt{|g|}\partial_{j}v),
\ee
is the Laplace--Beltrami operator over the Riemann surface $S$ with the metric $g=(g_{ij})$ ($i,j=1,2$) and $|g|=\det(g)$. Thus, with $v=\ln|u|^2$ and the data \eq{a4.3}, we see
that $v$ satisfies the sourceful equation
\begin{equation}\label{3.3}
\Delta v=4\e^v f(\e^v)-2+4\pi \sum_{s=1}^{M}\delta_{q_{s}}(x)- 4\pi\sum_{s=1}^{N}\delta_{p_{s}}(x), \quad x\in S,
\end{equation}
where $\delta_{p}(x)$ is the Dirac distribution concentrated at $p\in S$.

Conversely, if $v$ solves \eq{3.3}, it can be used to obtain a solution to the Bogomol'nyi equations \eq{2.13}
and \eq{2.14} (cf. \cite{Bra,N1,N2}) and such a construction may be made explicitly \cite{Ybook} as in the full-plane situation, that is, when $S=\bfR^2$, by the following two steps using a coordinate system.

\begin{enumerate}

\item[(i)] First, we set
\bea
u(z)&=&\mathrm{exp}\left(\frac{1}{2}v(z)+\mathrm{i}\theta(z)\right),\lb{a4.8}\\
 \theta(z)&=&\sum_{s=1}^{M}\mathrm{arg}(z-q_s)-\sum_{s=1}^{N}\mathrm{arg}(z-p_s),\lb{a4.9}
\eea

\item[(ii)] Then,  we  use the complex differentiation $\pa=\frac12(\pa_1-\ii\pa_2)$ to form
\be
A_1(z)=-\mathrm{Re}\{2\ii\overline{\pa}\ln u(z)\},\quad A_2(z)=-\mathrm{Im}\{2\ii\overline{\pa}\ln u(z)\}.\lb{a4.10}
\ee

\end{enumerate}

The formulas \eqref{a4.8}-\eqref{a4.10} allow us to represent the gauge-covariant derivatives in local coordinates as
\bea
D_1u&=&(\pa+\overline{\pa})u-\left(\frac{\overline{\pa}u}u-\frac{\pa\overline{u}}{\overline{u}}\right)u=u\pa v,\lb{x3.5}\\
D_2u&=&\ii(\pa-\overline{\pa})u+\ii\left(\frac{\overline{\pa}u}u+\frac{\pa\overline{u}}{\overline{u}}\right)u=\ii u\pa v.\lb{x3.6}
\eea
Consequently, in local isothermal coordinates, we have
\be\lb{aa4.13}
*(D_A u\wedge *\overline{D_A u})=\frac12|\nabla v|^2\e^v.
\ee
These formulas are useful and informative for us to calculate various quantities of interest in our study.

For example, we can show that the Hamiltonian energy density \eq{2.6} is regular at any pole of $u$. In fact, assume that the pole is at the origin. Then \eq{a4.2} and \eq{aa4.13} give us
\be\lb{aa4.14}
*(D_A u\wedge *\overline{D_A u})=m |z|^{-2(m+1)},\quad 0<|z|\ll1.
\ee
In view of \eq{2.12}, \eq{aa3.13}, \eq{a4.2}, and \eq{aa4.14}, we obtain
\be
F(|u|^2)*(D_A u\wedge *\overline{D_A u})=2mh'(0) |z|^{2(m-1)},\quad 0<|z|\ll1,
\ee
such that the expected regularity of the mid-term (for example) in the Hamiltonian density \eq{2.6} near the pole indeed follows for any $m\geq1$.

Moreover, it is clear that the zeros and poles of $u$ give rise to vortices and antivortices of a solution $(A,u)$ of \eqref{2.13} and\eqref{2.14}. In fact, in \eq{2.13}, the quantity
$*F_A$ represents the vorticity field. Since $w(s)$ is monotone decreasing and $w(0)=1$ and $w(\infty)=-1$, we see that the quantity $w(|u|^2)$ attains its global maximum value $1$
and minimum value $-1$ at the zeros and poles of $u$, respectively, indeed giving rise to vortices and antivortices as anticipated.

We will now obtain a vortex-antivortex solution to \eq{2.13} and \eq{2.14} through  a study of the equivalent nonlinear elliptic equation \eq{3.3}.

\section{Vortices and antivortices on a compact Riemann surface}\lb{sec5}

In this section, we establish an existence and uniqueness theorem for a multi-vortex-antivortex solution to the Bogomol'nyi equations \eqref{2.13} and\eqref{2.14}
on a compact Riemann surface
under a necessary and sufficient condition. This theorem may be stated as follows.

\begin{theorem}\label{thm1}
Consider the Bogomol'nyi equations \eqref{2.13} and \eqref{2.14} over a compact Riemann surface $S$ derived from the Hamiltonian energy density \eq{2.6} describing a generalized Abelian Higgs theory.
For any  points $q_{1},\dots,q_{M}$ and $p_{1},\dots,p_{N}$ on $S$, with repetition counting possible local
multiplicities, the equations  have a solution $(A,u)$ so that $q$'s are
 the zeros of $u$ and $p$'s are the poles of $u$, representing a prescribed distribution of vortices and antivortices, if and only if the condition
\begin{equation}\label{3.2}
|M-N|<\frac{\left|S\right|}{2 \pi}
\end{equation}
holds. Besides, if a solution exists, it is uniquely determined up to a
 gauge transformation. Furthermore, the total energy of the solution is quantized and related to the number of vortices $M$ and number of antivortices $N$ by the formula
\be\lb{axp5.2}
E=\int_S *{\cal H}(A,u)=2\pi(M+N).
\ee
Furthermore, these numbers are topological and give rise to the first Chern class and the Thom class following the expressions
\be\lb{axp5.3}
c_1=\frac1{2\pi}\int_S F_A=M-N,\quad \tau=\int_S \dd J=4\pi N,
\ee
respectively.
\end{theorem}

At this juncture, a few remarks are in order.

\begin{enumerate}

\item[(i)] The condition \eq{3.2} imposes an explicit upper bound on the net excess of either vortices and antivortices against their counterparts.
Mathematically, the bound confines the total net topological charge defined by the first Chern class, and, physically, 
 it means the system can only sustain a ``reasonably small" excess of vortices over antivortices, or vice versa, and larger imbalances would break
down the system, as measured by the associated total magnetic flux or vorticity charge,  as expressed by the first formula in \eq{axp5.3}.

\item[(ii)] The expression \eq{axp5.2} on the other hand indicates that there is no total-energy upper bound for a system of vortices and antivortices,
although thermodynamically, or energetically, of course, the system favors fewer vortices and antivortices, as shown by the Boltzmann partition function
\be
Z=\sum_{|M-N|<\frac{|S|}{2\pi}} \e^{-\frac{2\pi}{k T}(M+N)},
\ee
where $k$ is the Boltzmann constant and $T$ the absolute temperature. As a consequence, we may calculate the total internal energy of the system given by
\be
U=\frac{2\pi}Z\sum_{|M-N|<\frac{|S|}{2\pi}} (M+N)\,\e^{-\frac{2\pi}{k T}(M+N)},
\ee
and study the underlying thermodynamical properties of the system consisting of the microstates realized by the individual subsystems of vortices and antivortices in all possible combinations.

\item[(iii)] Note that the results on the coexisting vortices and antivortices described in Theorem \ref{thm1} and the above two remarks appear in the
{\em absence} of an external magnetic field. In other words, they appear spontaneously, as witnessed by the symmetry expressed by \eq{3.2} and \eq{axp5.2}, with respect to $M$ and $N$. However, such a symmetry will be broken when an external magnetic field, say $B$ (assuming $B$ being constant for simplicity),
is switched on. In fact, in this situation, the Hamiltonian energy density \eq{2.6} is modified into the form 
\be
{\cal H}_B(A,u)={\cal H}(A,u)-*F_A\, B,
\ee
such that, using \eq{axp5.2} and \eq{axp5.3}, the total energy reads
\be\lb{axp5.7}
E_B=\int_S *{\cal H}_B(A,u)=2\pi(M[1-B]+N[1+B]).
\ee
This expression leads to the following scenarios with respect to $B$ based on the least-energy principle:

\begin{enumerate}

\item[(a)] If $B$ is weak such that $|B|<1$, \eq{axp5.7} indicates that the system favors disappearance of vortices of any kind, that is, $M=N=0$.

\item[(b)] If $B$ is strong such that $|B|>1$, \eq{axp5.7} enables us to conclude that the system favors the appearance of one of the two kinds of vortices. Specifically, if $B>1$, the system favors as many vortices as possible but as few antivortices as possible, and if $B<-1$, the opposite
phenomenon occurs, indicating the fact that the system works to stay aligned with the external field and that the external field now breaks down the
spontaneous symmetry described in (i) and (ii) above.

\item[(c)] The parts (a) and (b) above establish the critical field
\be
B_c=1,
\ee
in the normalized model situation in our context, unveiling the onset of a type-II superconductivity mechanism in which $B_c$ corresponds to the first critical magnetic
field.

\item[(d)] Since the partition function is now given by
\be
Z=\sum_{|M-N|<\frac{|S|}{2\pi}} \e^{-\frac{2\pi}{k T}(M[1-B]+N[1+B])},
\ee
we see that the pictures depicted in (a)--(c) are what would happen at low temperature, $T\approx0$, when superconductivity takes place.

\end{enumerate}

\end{enumerate}

Below we prove this theorem. First, we show that \eq{3.2} is a necessary condition. Then, we use a fixed-point theory method to prove the existence of a solution under the condition
\eq{3.2}. The uniqueness of the
solution follows simply from the monotonicity of the function $sf(s)$ and the maximum principle.

\subsection{Necessity}
Using the background functions $v_{0}^{\prime}$ and $v_{0}^{\prime \prime}$ satisfying \cite{Aubin}
\begin{equation}\label{3.4}
\Delta v_{0}^{\prime}=-\frac{4 \pi M}{|S|}+4 \pi \sum_{s=1}^{M} \delta_{q_{s}},\quad \Delta v_{0}^{\prime \prime}=-\frac{4 \pi N}{|S|}+4 \pi \sum_{s=1}^{N} \delta_{p_{s}},
\end{equation}
and the substitution $v=v_{0}^{\prime}-v_{0}^{\prime \prime}+\varphi$ to remove the singular source terms,  we see that \eqref{3.3} becomes
\begin{equation}\label{3.6}
\Delta \varphi=4\re^{v^{\prime}_{0}-v^{\prime \prime}_{0}+\varphi}f(\re^{v^{\prime}_{0}-v^{\prime \prime}_{0}+\varphi})-2+\frac{4 \pi}{|S|}(M-N).
\end{equation}
Note that the representation \eq{aa3.13} implicates that the right-hand side of \eq{3.6} is regular at the poles, $p_1,\dots, p_N$.
On the other hand, in view of \eq{aa3.11} and \eq{aa3.14}, we have
\begin{equation}\label{3.7}
0\leq sf(s)\leq 1, \quad 0<s<\infty.
\end{equation}
Thus, integrating \eq{3.6} and using \eq{3.7}, we get
\be\lb{5.7}
0<|S|-2\pi (M-N)<2|S|,
\ee
which gives us \eq{3.2}.

\subsection{Sufficiency}

Now we assume that \eqref{3.2} holds and show that equation \eqref{3.6} has a solution. We shall use a Leray--Schauder  fixed-point theory argument over the Sobolev space  $W^{1,2}(S)$.  To proceed, we define
\begin{equation}\label{3.13}
X=\left\{\varphi \in W^{1,2}(S) \bigg| \int_{S} \varphi \mathrm{~d} \sigma=0\right\},
\end{equation}
where $\dd\sigma$ is the area element of the surface $(S,g)$ so that we have the direct sum  $W^{1,2}(S)=\bfR \oplus X$  and  the Poincar\'{e} inequality
\begin{equation}\label{3.14}
 \int_{S} \varphi^{2} \mathrm{~d} \sigma \leq  C_{1} \int_{S}|\nabla \varphi|^{2} \mathrm{~d} \sigma, \quad \varphi \in X,
\end{equation}
 where in local coordinates, $|\nabla \varphi|^2= g^{j k} \partial_{j} \varphi \partial_{k} \varphi$,  and $C_1$ is a positive constant whose value is of
no concern in our analysis of the problem.

We rewrite the equation \eq{3.6} as
\begin{equation}\lb{5.11}
\Delta \varphi=4\re^{v^{\prime}_{0}-v^{\prime \prime}_{0}+c+\varphi}f(\re^{v^{\prime}_{0}-v^{\prime \prime}_{0}+c+\varphi})-C_0,\quad C_0=2-\frac{4 \pi}{|S|}(M-N),\quad \varphi\in X,
\end{equation}
where $c\in\bfR$ is a constant depending on $\vp$. Integrating \eq{5.11}, we have
\begin{equation}\label{3.16}
\int_{S} \mathrm{e}^{v^{\prime}_{0}-v^{\prime \prime}_{0}+\varphi+c}f(\mathrm{e}^{v^{\prime}_{0}-v^{\prime \prime}_{0}+\varphi+c}) \mathrm{d} \sigma=\frac14C_{0}|S|.
\end{equation}

\begin{lemma}\label{lem1}
For given  $\varphi \in X $, there is a unique number $c=c(\varphi) \in \mathbb{R}$ such that \eq{3.16} is fulfilled, provided that the condition \eq{3.2} is observed.
\end{lemma}

\begin{proof}
Consider the function
\begin{equation}
g(c)= \int_{S} \mathrm{e}^{v^{\prime}_{0}-v^{\prime \prime}_{0}+\varphi+c}f(\mathrm{e}^{v^{\prime}_{0}-v^{\prime \prime}_{0}+\varphi+c}) \mathrm{d} \sigma.
\end{equation}
By the bounded convergence theorem, we have
\begin{equation}
\lim _{c \rightarrow-\infty} g(c)=0, \quad \lim _{c \rightarrow \infty} g(c)=\left|S\right|.
\end{equation}
On the other hand, \eq{3.2} or \eq{5.7} implies $0<C_0<4$. Since the bounded convergence theorem also implies
 the continuity of  $g(c)$, we get that there is a point $c$ so that $ g(c)=\frac14C_{0}|S|$. Besides, since
\eq{aa3.11} implies $g(c)$ is monotone increasing for $c\in\bfR$, we see that for given $\vp\in X$ the solution to $g(c)=\frac14C_0|S|$ is unique. This solution may be denoted as $c(\vp)$ such that
the proof of the lemma follows.
\end{proof}
\begin{lemma}\label{lem2}
 For given $\varphi \in X$, let $c(\varphi)$ be defined as in Lemma \ref{lem1}. Then, viewed as a function, $c: X \rightarrow \mathbb{R}$  is continuous with respect to the weak topology of $X$.
\end{lemma}
\begin{proof}
Assume that $\left\{\varphi_{n}\right\}$ is a sequence in $X$ such that
$
\varphi_{n} \to  \varphi_{0} \in X   \text{ weakly as }  n \rightarrow \infty.
$
Using the compact embedding
$
 W^{1,2}(S) \rightarrow L^{p}(S)  \text{ for }~ p \geq 1,
$
 we obtain
 $
 \varphi_{n} \rightarrow \varphi_{0}$  strongly in  $L^{p}(S)$.
We need to show that  $c\left(\varphi_{n}\right) \rightarrow c\left(\varphi_{0}\right)$  as  $n \rightarrow \infty$.
This can be accomplished through a few steps as follows.

Step 1. The sequence  $\left\{c\left(\varphi_{n}\right)\right\} $ is bounded from above.

Otherwise, extracting a subsequence if necessary, we may assume that  $c\left(\varphi_{n}\right) \rightarrow \infty$ as $n \rightarrow \infty$. By the strong convergence  $\varphi_{n} \rightarrow \varphi_{0}$  in  $L^{p}(S)$ and the Egorov theorem, we obtain that for any $ \varepsilon>0 $ there is a sufficiently large number  $K_{\varepsilon}>0$  and a subset $S_{\varepsilon} \subset S$  such that
\begin{equation}\label{3.18}
\left|\varphi_{n}(x)\right| \leq K_{\varepsilon}, \quad x \in S-S_{\varepsilon}; \quad\left|S_{\varepsilon}\right|<\varepsilon .
\end{equation}
As a result, by replacing $\varphi$ with $\varphi_{n}$ in equation \eqref{3.16} and applying \eqref{3.18}, we arrive at
\ber\label{3.19}
\frac14 C_{0}\left|S\right|&=&\int_{S} \mathrm{e}^{v^{\prime}_{0}-v^{\prime \prime}_{0}+\varphi_n+c(\varphi_n)}f(\mathrm{e}^{v^{\prime}_{0}-v^{\prime \prime}_{0}+\varphi_n+c(\varphi_n)})  \dd \sigma \nn\\
&\geq & \int_{S-S_\varepsilon}\mathrm{e}^{v^{\prime}_{0}-v^{\prime \prime}_{0}+\varphi_n+c(\varphi_n)}f(\mathrm{e}^{v^{\prime}_{0}-v^{\prime \prime}_{0}+\varphi_n+c(\varphi_n)})  \dd \sigma\nn \\
&\geq& \int_{S-S_\varepsilon}\mathrm{e}^{v^{\prime}_{0}-v^{\prime \prime}_{0}+c(\varphi_n)-K_\varepsilon}f(\mathrm{e}^{v^{\prime}_{0}-v^{\prime \prime}_{0}+c(\varphi_n)-K_\varepsilon })\dd \sigma.
\eer
Here, we have used the property $0\leq s f(s)\leq1$ and the monotonicity of $sf(s)$.
Taking $n \rightarrow \infty$ on the right-hand side of the above inequality, we arrive at
$
C_{0}\left|S\right|\geq 4(\left|S\right|-\varepsilon).
$
Since $\varepsilon>0$ is arbitrary, we get $C_0\geq 4$, contradicting the condition $C_{0}<4$.

Step 2. The sequence $\left\{c\left(\varphi_{n}\right)\right\} $ is bounded from below.

The proof is similar to that in Step 1. In fact, assume otherwise that $c(\vp_n)\to-\infty$ as $n\to\infty$. Using \eq{3.18} and replacing $\vp$ in \eq{3.16} by $\vp_n$, we get
\bea
\frac14C_0|S|&=&\left(\int_{S-S_\vep}+\int_{S_\vep}\right)\mathrm{e}^{v^{\prime}_{0}-v^{\prime \prime}_{0}+\varphi_n+c(\varphi_n)}f(\mathrm{e}^{v^{\prime}_{0}-v^{\prime \prime}_{0}+\varphi_n+c(\varphi_n)})  \dd \sigma \nn\\
&\leq&\int_{S-S_\vep}\mathrm{e}^{v^{\prime}_{0}-v^{\prime \prime}_{0}+K_\vep+c(\varphi_n)}f(\mathrm{e}^{v^{\prime}_{0}-v^{\prime \prime}_{0}+K_\vep+c(\varphi_n)})  \dd \sigma +|S_\vep|.\lb{5.20}
\eea
Taking $n\to\infty$ in \eq{5.20}, we get $C_0|S|\leq 4\vep$. Since $\vep>0$ is arbitrary, we arrive at $C_0\leq0$,  contradicting $C_0>0$.

Step 3. $c(\vp_n)\to c(\vp_0)$ as $n\to\infty$.

In fact, since $\left\{c\left(\varphi_{n}\right)\right\}$ is bounded, we may extract a subsequence if necessary such that we may assume $c\left(\varphi_{n}\right) \rightarrow$ some $c_0 \in \mathbb{R}$ as $n \rightarrow \infty$. It suffices to show $c_0=c(\vp_0)$.

For any $\vep>0$, let $S_\vep\subset S$ be such that
$\vp_n\to\vp_0$ uniformly on $S-S_\vep$ and $|S_\vep|<\vep$. Then we have
\bea
&&\left(\int_{S_\vep}+\int_{S-S_\vep}\right) \mathrm{e}^{v^{\prime}_{0}-v^{\prime \prime}_{0}+\varphi_0+c_0}f(\mathrm{e}^{v^{\prime}_{0}-v^{\prime \prime}_{0}+\varphi_0+c_0}) \mathrm{d} \sigma\nn\\
&&\leq \vep+\lim_{n\to\infty}\int_{S-S_\vep}\mathrm{e}^{v^{\prime}_{0}-v^{\prime \prime}_{0}+\varphi_n+c(\vp_n)}f(\mathrm{e}^{v^{\prime}_{0}-v^{\prime \prime}_{0}+\varphi_n+c(\vp_n)}) \mathrm{d} \sigma\nn\\
&&\leq \vep+\frac14C_0|S|.
\eea
Since $\vep>0$ is arbitrary, we obtain
\be\lb{5.22}
\int_{S} \mathrm{e}^{v^{\prime}_{0}-v^{\prime \prime}_{0}+\varphi_0+c_0}f(\mathrm{e}^{v^{\prime}_{0}-v^{\prime \prime}_{0}+\varphi_0+c_0}) \mathrm{d} \sigma\leq \frac14C_0|S|.
\ee
Similarly, we have
\bea
&&\int_{S} \mathrm{e}^{v^{\prime}_{0}-v^{\prime \prime}_{0}+\varphi_0+c_0}f(\mathrm{e}^{v^{\prime}_{0}-v^{\prime \prime}_{0}+\varphi_0+c_0}) \mathrm{d} \sigma\nn\\
&&\geq \lim_{n\to\infty}\int_{S-S_\vep}\mathrm{e}^{v^{\prime}_{0}-v^{\prime \prime}_{0}+\varphi_n+c(\vp_n)}f(\mathrm{e}^{v^{\prime}_{0}-v^{\prime \prime}_{0}+\varphi_n+c(\vp_n)}) \mathrm{d} \sigma\nn\\
&&\geq \lim_{n\to\infty}\left(\int_{S_\vep}+\int_{S-S_\vep}\right)\mathrm{e}^{v^{\prime}_{0}-v^{\prime \prime}_{0}+\varphi_n+c(\vp_n)}f(\mathrm{e}^{v^{\prime}_{0}-v^{\prime \prime}_{0}+\varphi_n+c(\vp_n)}) \mathrm{d} \sigma-\vep\nn\\
&&=\frac14C_0|S|-\vep.
\eea
Since $\vep>0$ is arbitrary, we have reversed the inequality \eq{5.22}. Therefore equality in \eq{5.22} holds.

By Lemma \ref{lem1}, we derive $c_0=c(\varphi_0)$. Thus $c(\varphi_n)\to c(\varphi_0)$ as $n\to\infty$  as claimed.
\end{proof}

 We are now ready to use the fixed-point method to obtain a solution of the equation \eqref{3.6}.
 For this purpose,  we pick $\varphi \in X$ and consider the equation
 \begin{equation}\label{3.24}
     \Delta \psi=4\re^{v^{\prime}_{0}-v^{\prime \prime}_{0}+\varphi+c(\varphi)}f(\re^{v^{\prime}_{0}-v^{\prime \prime}_{0}+\varphi+c(\varphi)})-C_{0}.
 \end{equation}
By \eqref{3.16}, the right-hand side of \eqref{3.24} has zero integral over $S$. Thus, 
the equation \eqref{3.24} has a unique solution $\psi$ in $X$ such that the relation $\varphi \to \psi$ defines a map $T:X\to X$
with $T(\vp)=\psi$.

\begin{lemma}\label{lem3a}
The map $T: X \to X$ is completely continuous.
\end{lemma}
\begin{proof}
Let $\left\{\varphi_{n}\right\}$ be a sequence in $X$ that satisfies
$
\varphi_{n} \to \varphi_{0}$ weakly in $X$  as $n \to \infty.
$
So we have
$
\varphi_{n} \to \varphi_{0}$ strongly in $L^{p}(S)$ ($p \geq 1$).

Let $\psi_{n}=T(\varphi_{n})$ and $\psi_{0}=T(\varphi_{0})$, then
\be\label{3.25}
\Delta (\psi_{n}-\psi_{0})=h(v^{\prime}_{0}-v^{\prime \prime}_{0}+\varphi_{n}+c(\varphi_{n}))-h(v^{\prime}_{0}-v^{\prime \prime}_{0}+\varphi_{0}+c(\varphi_{0})),
\ee
where $h(v)=4\re^v f(\re^v)$. We are to show that $\psi_n\to\psi_0$ in $X$.
Multiplying  \eqref{3.25} by $\psi_n-\psi_0$ and integrating by parts, we get
\bea\lb{5.23}
&&\int_{S}\left|\nabla(\psi_{n}-\psi_{0})\right|^{2}\mathrm{d}\sigma\nn\\
&&\leq\|h(v^{\prime}_{0}-v^{\prime \prime}_{0}+\varphi_{n}+c(\varphi_{n}))-h(v^{\prime}_{0}-v^{\prime \prime}_{0}+\varphi_{0}+c(\varphi_{0}))\|_{L^2}\|\psi_n-\psi_0\|_{L^2}\nn\\
&&\leq \vep\|\psi_n-\psi_0\|^2_{L^2}+\frac1{4\vep} \|h(v^{\prime}_{0}-v^{\prime \prime}_{0}+\varphi_{n}+c(\varphi_{n}))-h(v^{\prime}_{0}-v^{\prime \prime}_{0}+\varphi_{0}+c(\varphi_{0}))\|_{L^2}^2,\quad\quad
\eea
where $\vep>0$ is arbitrary.  Using \eq{3.14} in \eq{5.23}, we have
\be\lb{5.24}
\|\psi_n-\psi_0\|_{W^{1,2}}\leq C\|h(v^{\prime}_{0}-v^{\prime \prime}_{0}+\varphi_{n}+c(\varphi_{n}))-h(v^{\prime}_{0}-v^{\prime \prime}_{0}+\varphi_{0}+c(\varphi_{0}))\|_{L^2}
\ee
for some constant $C>0$. Since $h(v)$ is bounded, $\vp_n\to\vp_0$ strongly in $L^p(S)$ ($p\geq1$), and $c(\vp_n)\to c(\vp_0)$, as $n\to\infty$, we can use the Egorov theorem to see that the right-hand side of
\eq{5.24} tends to zero as $n\to\infty$. Hence $\psi_n\to\psi_0$ in $W^{1,2}(S)$ as claimed.
\end{proof}
We now establish {\em a priori} estimate for the fixed points.
\begin{equation}\label{3.26}
    \varphi_{t}=tT(\varphi_{t}),\quad 0\le t \le 1.
\end{equation}
\begin{lemma}\label{lem4}
 There is a number $C>0$ independent of $t \in\left [ 0,1 \right ] $ so that
\begin{equation}\label{3.27}
 \left \| \varphi_{t} \right \|_{X} \le C,\quad 0 \le t \le 1.
\end{equation}
As a result, $T$ has a fixed point in $X$.
\end{lemma}
\begin{proof}
According to the definition \eqref{3.26},  we see that for $t>0$, $\varphi_{t}$ satisfies the equation
\begin{equation}\label{3.28}
    \Delta \varphi_{t}=
4t\re^{v^{\prime}_{0}-v^{\prime \prime}_{0}+c(\varphi_{t})+\varphi_{t}}f(\mathrm{e}^{v^{\prime}_{0}-v^{\prime \prime}_{0}+c(\varphi_{t})+\varphi_{t}})-C_{0}t.
\end{equation}
Multiplying both sides of \eqref{3.28} by  $\varphi_{t}$  and integrating by parts, we arrive at
\bea\lb{axp5.27}
\int_{S}\left|\nabla \varphi_{t}\right|^{2} \dd\sigma
&\leq&  \int_{S}\left|4t\re^{v^{\prime}_{0}-v^{\prime \prime}_{0}+c(\varphi_{t})+\varphi_{t}}f(\mathrm{e}^{v^{\prime}_{0}-v^{\prime \prime}_{0}+c(\varphi_{t})+\varphi_{t}})\varphi_t\right| \dd\sigma\nn\\
&\leq& 4|S|^{\frac12}\|\vp_t\|_{L^2}\leq\vep\|\varphi\|^2_{L^2}+\frac4\vep|S|,
\eea
where $\vep>0$ is an arbitrary constant.
Using \eqref{3.14} on the right-hand side of \eq{axp5.27} and choosing $\vep$ to be suitably small,  we obtain the boundedness of the left-hand
side of \eq{axp5.27} and we thus establish the lemma.
\end{proof}
\begin{lemma}
The solution to \eq{3.6} is unique.
\end{lemma}
\begin{proof}
In fact, let $\vp$ and $\psi$ be two solutions to \eq{3.6}. Then we have
\bea
\Delta(\vp-\psi)&=&4\re^{v^{\prime}_{0}-v^{\prime \prime}_{0}+\varphi}f(\re^{v^{\prime}_{0}-v^{\prime \prime}_{0}+\varphi})-4\re^{v^{\prime}_{0}-v^{\prime \prime}_{0}+\psi}f(\re^{v^{\prime}_{0}-v^{\prime \prime}_{0}+\psi})\nn\\
&=& h'(v^{\prime}_{0}-v^{\prime \prime}_{0}+\tilde{\varphi})(\vp-\psi),
\eea
where $\tilde{\vp}$ is between $\vp$ and $\psi$ and $h(v)=4\e^v f(\e^v)$. Since $h'(v)\geq0$, we see that the maximum principle implies that $\vp=\psi$ on $S$.
\end{proof}

\subsection{Topological numbers}

For simplicity and without loss of generality, we consider the plus sign situation only (that is, we choose to consider the upper sign, $+$, situation
in \eq{2.13} and \eq{2.14}).

First, in view of \eq{2.13}, we see that the first Chern class is
\bea\lb{5.30}
c_1&=&\frac1{2\pi}\int_S F_A=\frac1{2\pi}\int_S *F_A\,\dd\sigma\nn\\
&=&\frac1{2\pi}\int_S w(|u|^2)\,\dd\sigma\nn\\
&=&\frac1{2\pi}\int\left(1-2|u|^2 f(|u|^2)\right)\,\dd\sigma\nn\\
&=&\frac1{2\pi}\int\left(1-2\re^{v^{\prime}_{0}-v^{\prime \prime}_{0}+\varphi}f(\re^{v^{\prime}_{0}-v^{\prime \prime}_{0}+\varphi})\right)\,\dd\sigma\nn\\
&=&(M-N),
\eea
in view of \eq{2.11} and  \eq{3.6}.

Next, we calculate the quantity $\int_S \dd J$, where $\dd J$ is given in \eq{2.10}, which is known to give rise to the Thom class \cite{SSY}
for the special case \eq{a3.11}.  Since it can be examined that only the poles of $u$ contribute to
this quantity, we can confine our study on a local region of a pole $p$ of $u$.  Thus, with local isothermal coordinates around the poles $p_1,\dots, p_N$ of $u$, with repetitions allowed to count for multiplicities, we have 
\bea\lb{5.31}
\tau&=&\int_S \dd J=\int_S (*\dd J)\,\dd\sigma\nn\\
&=&\sum_{i=1}^N \lim_{r\to 0}\int_{|x-p_i|\leq r}J_{12}\dd x\nn\\
&=&
\sum_{i=1}^N \lim_{r\to 0}\oint_{|x-p_i|=r}(J_1 \dd x_1+J_2 \dd x_2)\equiv\sum_{i=1}^N \tau(p_i).
\eea
On the other hand,
inserting \eq{x3.5} and \eq{x3.6} into \eq{2.9}, we have
\bea\lb{x3.7}
\oint_{|x-p|=r}(J_1 \dd x_1+J_2 \dd x_2)&=&\ri\oint_{|x-p|=r}|u|^2 f(|u|^2)([\overline{\pa}-\pa] v\dd x_1-\ii[\overline{\pa}+\pa]v\dd x_2)\nn\\
&=&\oint_{|x-p|=r}|u|^2 f(|u|^2)(-\pa_2 v\dd x_1+\pa_1 v\dd x_2).
\eea
Let $p$ be a pole of $u$. Then, in view of the local representation \eq{a4.2}, we have
\be\lb{x3.8}
|u(x)|^2=|x-p|^{-2m}\gamma(x)=r^{-2m}\gamma(x),\quad r=|x-p|\ll1,
\ee
where $\gamma(x)$ is a smooth nonvanishing function near $x=p$.
Substituting \eq{x3.8} into \eq{x3.7} and using the polar coordinates $x_1=r\cos\theta, x_2=r\sin\theta$, we obtain
\be\lb{x3.9}
\tau(p)=\lim_{r\to0}\int^{-2\pi}_0|u|^2 f(|u|^2)r\pa_r v\,\dd\theta=4\pi m,
\ee
by virtue of \eq{aa3.14}.
In view of \eq{x3.9}, we see that \eq{5.31} renders us the result
\be\lb{5.35}
\tau=\sum_{i=1}^N \tau(p_i)=4\pi N.
\ee

With the results \eq{5.30} and \eq{5.35}, we can integrate \eq{a3.12} to get
\be
E=\int_S {\cal H}(A,u)\,\dd\sigma=2\pi (M+N)
\ee
as asserted.

The proof of Theorem \ref{thm1} is complete.

\section{Existence and uniqueness theorem on the full plane}

From the condition \eqref{3.2}, we observe that a surface with a larger volume can support a greater range of the difference between the number of vortices $M$ and the number of antivortices $N$. Especially, this implies that for a surface with infinite volume, the two numbers $M$ and $N$ may be arbitrary. In this section, we shall present such a result on $\mathbb{R}^{2}$  under the technical conditions
\bea
&&4s f(s)\leq 1+s,\quad s\in [0,1],\lb{C1}\\
&&sf(s)+\frac1sf\left(\frac1s\right)\geq1,\quad s\in (0,1].\lb{C2}
\eea
It is clear that these conditions cover the classical model \eq{a3.11} as a special case. In fact, for \eq{a3.11}, the condition \eq{C1} holds as an inequality and the condition \eq{C2} as an equality.

Our existence and uniqueness results for the solutions to the Bogomol'nyi equations \eqref{2.13} and \eqref{2.14} over the full plane, assuming the conditions \eq{C1} and \eq{C2}, are stated as follows.

\begin{theorem}\label{thm2}
Given any $M$  points  $q_{1}, \ldots, q_{M}$  and  $N$  points  $p_{1}, \ldots, p_{N}$ in $\mathbb{R}^{2}$, the Bogomol'nyi equations \eqref{2.13} and \eqref{2.14}
governed by the Hamiltonian density \eq{2.6} over $\mathbb{R}^{2}$ where the coupling functions $w$ and $F$ are related through \eq{2.11} have a unique  finite-energy solution $(A, u)$,
up to gauge-transformation equivalence, representing $M$ vortices at the points $q$'s and $N$ antivortices at the points $p$'s so that the $q$'s and $p$'s are the
zeros and poles of $u$, respectively, with algebraic multiplicities.  As $|x| \to \infty$, the solution approaches the vacuum state  with spontaneously broken symmetry characterized by
$F_{12}=0$ and $|u|=1$. Under the condition  $F(1)>0$, these limits are achieved exponentially fast at the rate
\be
|u|^2-1,\quad \left|D_{j} u\right|, \quad F_{12}={\rm{O}}(\e^{-\sqrt{F(1)}(1-\vep)|x|}), \quad |x|\gg1,
\ee
where $\vep>0$ is an arbitrarily small number, and
the total magnetic flux or the Chern-class charge, the Thom-class charge, and minimal energy have the  quantized values,
\be
\int_{\bfR^{2}} F_{12}\,\dd x=2 \pi(M-N), \quad \int_{\bfR^2}J_{12}\,\dd x=4\pi N, \quad E=\int_{\bfR^2}{\cal H}\,\dd x=2\pi (N+M),
\ee
respectively.
\end{theorem}

Note that the conditions \eq{C1} and \eq{C2} are imposed for convenience of proof with attention that they accommodate some familiar nonlinear coupling functions, including a few
 classical situations. These conditions may be replaced by other conditions which may require more effort in solution construction. In the next section, we will comment on some of these issues.

The proof of Theorem \ref{thm2} will be centered around the equation \eq{3.3} over $\bfR^2$ or
\begin{equation}\label{a6.5}
\Delta v=4\e^v f(\e^v)-2+ 4\pi\sum_{s=1}^{M}\delta_{q_{s}}(x)-4\pi \sum_{s=1}^{N}\delta_{p_{s}}(x), \quad x\in \bfR^2,
\end{equation}
subject to the boundary condition
\be\lb{a6.6}
v(x)\to0,\quad |x|\to\infty.
\ee

Since the domain is now the full plane such that there is no more restriction that confines the numbers of vortices and antivortices in terms of the size of the domain, the treatment of the problem eases greatly, yet rendering strong existence and unique results under general conditions, as stated in the theorem. To see this, we proceed as follows.

Consider the Taubes equation \cite{Gud,Manton,T1}
\be\lb{a6.7}
\Delta v =(\e^v-1)+4\pi\sum_{s=1}^N \delta_{p_s}(x),\quad x\in\bfR^2.
\ee
It has been shown \cite{JT,T1} that \eq{a6.7} has a unique solution subject to \eq{a6.6} which satisfies $v<0$ everywhere. Use $v_1$ to denote such a solution. Then by \eq{C1} we have 
\bea
\Delta(-v_1)&=&(1-\e^{v_1})-4\pi\sum_{s=1}^N\delta_{p_s}\nn\\
&\leq&2-4\e^{v_1} f(\e^{v_1})-4\pi\sum_{s=1}^N\delta_{p_s}\nn\\
&\leq&4\e^{-v_1}f(\e^{-v_1})-2-4\pi\sum_{s=1}^N\delta_{p_s}+ 4\pi\sum_{s=1}^{M}\delta_{q_{s}},
\eea
using \eq{C2} in the last step. This establishes $v_+=-v_1$ as a positive supersolution to \eq{a6.5} subject to \eq{a6.6}.

To obtain a subsolution in the same manner, we again consider \eq{a6.7} subject to \eq{a6.6} with $p$'s being replaced by $q$'s and $N$ by $M$ as given in \eq{a6.5}, and use
$v_2$ to denote such a solution. Then $v_2$ satisfies
\bea
\Delta v_2&=&(\e^{v_2}-1)+4\pi\sum_{s=1}^M \delta_{q_s}\nn\\
&\geq&4\e^{v_2}f(\e^{v_2})-2 +4\pi\sum_{s=1}^M \delta_{q_s}-4\pi\sum_{s=1}^N\delta_{p_s},
\eea
using \eq{C1} alone. That is, $v_-=v_2$ is a negative subsolution to \eq{a6.5} subject to the boundary condition \eq{a6.6}.

As a consequence of the theory of nonlinear elliptic differential equations, we know that \eq{a6.5} subject to \eq{a6.6} has a solution $v$ satisfying
\be
v_-<v<v_+
\ee
everywhere in $\bfR^2$. By the monotonicity condition \eq{aa3.11}, we know that such a solution must be unique.

Using \eq{a6.6}, we see that the solution $v$ and any of its partial derivative $\pa_j v$ ($j=1,2$) all satisfy the linearized equation
\be
\Delta W=4(f(1)+f'(1))=F(1) W,
\ee
 near infinity of $\bfR^2$ in view of the condition \eq{2.11}. Thus, assuming $F(1)>0$ and using a standard maximum principle argument, we can obtain the sharp exponential asymptotic estimate
for $W$:
\be
|W(x)|\leq C(\vep)\e^{-(1-\vep)\sqrt{F(1)}|x|},\quad |x|>R,
\ee
where $R>0$ is sufficiently large, $\vep>0$ may be arbitrarily small, and $C(\vep)>0$ is a constant depending on $\vep$. Applying this estimate to $v$ and its derivatives, we get
\begin{equation}\lb{6.13}
\left|v(x)\right|\leq C(\vep)\e^{-(1-\vep)\sqrt{F(1)}\left|x\right|}, \quad \left|\nabla v(x)\right|\leq C(\vep) \e^{-(1-\vep)\sqrt{F(1)}|x|},\quad |x|>R.
\end{equation}

In view of \eq{2.12}, \eq{2.13}, \eq{x3.5},  \eq{x3.6}, and \eq{6.13}, we have
\be\lb{6.14}
|u|^2-1,\ \ 
|D_j u|,\ \  j=1,2, \ \  F_{12}=\mbox{O}(\e^{-(1-\vep)\sqrt{F(1)}|x|}).
\ee

With \eq{6.14}, we can calculate various quantities as in the compact situation since the boundary terms will not make contribution. 

Thus the theorem is established.

\medskip

There are plenty of models covered by the conditions \eq{C1} and \eq{C2} which are of potential interest for phenomenological applications. For example, here is one:
\be\lb{6.15}
w(s)=\frac{1-s^m}{1+s^m},\quad F(s)=\frac{4m s^{m-1}}{(1+s^m)^2}, \quad m=1,2,\dots.
\ee
The link function $f(s)$ defined in \eq{2.11} leads to
\be\lb{6.16}
sf(s)=\frac{s^m}{1+s^m}.
\ee
With \eq{6.16}, we can examine to see that \eq{C1} is satisfied as an inequality and \eq{C2} is satisfied as an equality again. Inserting \eq{6.15} into \eq{2.6}, we obtain the Abelian 
Higgs theory
\be\lb{6.17}
{\cal H}(A,u)=\frac12 *(F_A\wedge *F_A)+\frac{2m|u|^{2(m-1)}}{(1+|u|^{2m})^2}*(D_A u\wedge *\overline{D_A u})+\frac12\left(\frac{1-|u|^{2m}}{1+|u|^{2m}}\right)^2,
\ee
with a Bogomol'nyi structure such that its vortex equation \eq{2.13} reads
\be\lb{6.18}
*F_A= \pm \frac{1-|u|^{2m}}{1+|u|^{2m}}.
\ee
It is clear that we may choose the integer $m$ to be arbitrarily large to make the vorticity field $*F_A$ in \eq{6.18} be as locally concentrated as we please around the zeros and poles of $u$.
Such a mechanism may be used in the setting of cosmic strings to render high curvature lumps for the gravitational sector governed by the Einstein equations \cite{Y1,Y2,Ybook,Y4} in the study
of galaxy formation problem in the early universe \cite{Gar,K,V,VS}.

\section{Remarks}
\setcounter{equation}{0}

We note that the condition \eq{C1} is imposed in order to use the solution to \eq{a6.7} as a subsolution to \eq{a6.5}. This condition can actually be removed but some additional effort has to be made
such that a subsolution can be constructed directly.  Here we omit these technical details.

Moreover, the conditions \eq{C1} and \eq{C2} may also be replaced with some other conditions to ensure the validity of Theorem \ref{thm2}. For example, using a direct minimization
method as in \cite{HHY}, the same results hold under the condition
\be\lb{7.1}
\int_0^v \left(2\e^sf(\e^s)-1\right)\,\dd s\geq \lm \ln\left(\frac{\cosh v+1}2\right),
\ee
where $\lm>0$ is some constant. It can be examined that the classical case \eq{a3.11} corresponds to $\lm=1$ in \eq{7.1}.

Recall that 
 the Gauss energy, or nonlinear sigma model energy,
\be
E(u)=\int_{\bfR^2} J |\nabla u|^2 \, \dd x,
\ee
is the foundational Hamiltonian for describing topological excitations  in the forms of vortices and antivortices in the 2D XY model, where $u(x) \in S^1 \subset \mathbb{C}$ and $J>0$ is a constant.
It arises as the continuum limit of the lattice model with nearest-neighbor spin interaction, known as the Ising model, and vortex configurations minimize this energy subject to topological constraints.
Although in the standard XY model, the modulus of the spin field is constrained, $|u|=1$,  but in generalized or relaxed models, such as the Ginzburg--Landau-type extensions, $|u|$ may deviate from unity, allowing for amplitude fluctuations. This makes the model closer to the Abelian Higgs model, or the complex Ginzburg--Landau model, where both the phase and amplitude are dynamical and affect vortex structure. Thus, the extended Hamiltonian
\be
E_F(u) = \int_{\mathbb{R}^2} J F(|u|^2) |\nabla u|^2 \, \dd x,
\ee
with some positive function $F: \mathbb{R}_+ \to \mathbb{R}_+$
 introduces amplitude-dependent modulation of the energy such that $F$ enriches the model in several aspects including (i)
If $F(s)\to\infty$ 
as $s\to0$, or $F(s)\to F_0$ as $s\to0$ and $F_0>F(s)$ for any $s>0$, then the energy penalizes vanishing amplitude or suppressing singular vortex cores.
(ii)
If $F(s)$ increases with $s$, the system penalizes large amplitudes.
(iii)
If $F(s)\sim1$, the model approaches the classical XY case. (iv) In the classical XY model, the vortices are singular -- the energy density diverges logarithmically at the core. In models with amplitude freedom and nontrivial $F(|u|^2)$, the vortex cores can become regularized, i.e., the field $u$ can go to zero at the core, softening the singularity.

In view of these,
we see that $F$ can serve to interpolate a wide variety of phenomena. 
It is in this background that the Hamiltonian \eq{2.6} arises
naturally and relevantly.

\medskip

The authors would like to thank an anonymous referee whose comments and suggestions help improve the presentation of this work.

\medskip

 {\bf Declaration of interests.} The authors declare that they have no known competing financial interests or personal relationships that could have appeared to influence the work reported in this paper.

\end{document}